\documentclass{article}
\usepackage{amsmath}
\usepackage{amssymb}
\usepackage{mathdots}

\makeatletter

\usepackage{amsthm}
\usepackage{cases}

\theoremstyle{plain}
\newtheorem{thm}{Theorem}
\numberwithin{thm}{section}
\newtheorem{prop}{Proposition}
\numberwithin{prop}{section}
\newtheorem{lem}{Lemma}
\numberwithin{lem}{section}

\theoremstyle{definition}
\newtheorem{remark}{Remark}
\numberwithin{remark}{section}
\newtheorem{exa}{Example}
\numberwithin{exa}{section}

\author{Han Feng\thanks{Department of Statistics, University of Warwick, Coventry, CV4 7AL, UK. H.Feng@warwick.ac.uk} 
 \and David Hobson\thanks{Department of Statistics, University of Warwick, Coventry, CV4 7AL, UK. D.Hobson@warwick.ac.uk}}

\makeatother

\usepackage[english]{babel}

\begin{document}

\title{\textbf{Gambling in Contests with Regret}}
\maketitle
\begin{abstract}
This paper discusses the gambling contest introduced in Seel \& Strack~\cite{GamblinginContests}
and considers the impact of adding a penalty associated with failure
to follow a winning strategy.

The Seel \& Strack model consists of $n$-agents each of whom privately
observes a transient diffusion process and chooses when to stop it.
The player with the highest stopped value wins the contest, and each
player's objective is to maximise their probability of winning the
contest. We give a new derivation of the results of Seel \& Strack~\cite{GamblinginContests}
based on a Lagrangian approach. Moreover, we consider an extension
of the problem in which in the case when an agent is penalised when
their strategy is suboptimal, in the sense that they do not win the
contest, but there existed an alternative strategy which would have
resulted in victory. 
\end{abstract}

\section{Introduction}

In \cite{GamblinginContests}, Seel \& Strack introduced a model of
a gambling contest between agents in which the objective of each agent
is not to maximise her return, but rather to maximise the probability
that her return is highest amongst the set of agents. One motivation
for studying such a problem is that it provides a stylised model for
competing fund managers, only the most successful of whom will be
given funds to invest over the next time period. Another distinct
strand of the literature on modelling competition between fund managers
is represented by Basak and Makarov~\cite{Basak}.

In the Seel \& Strack paper each agent observes a stochastic process
and chooses a stopping time to produce a stopped value. The agent
wins the contest if her stopped value is greater than the stopped
value of every other agent, and the objective of the agent is to maximise
the probability that she wins the contest. Our aim is to add a behavioural
finance aspect to the problem. Again the objective of the agent is
to maximise her chances of winning the contest, but now she is penalised
if she has not won the contest, and she has behaved sub-optimally,
in the sense that there was an alternative strategy which would have
led to her winning the contest. Thus a fund manager who has followed
a poor strategy is not merely given a new role within the firm, but
instead is terminated with disgrace.

Although the problem described in \cite{GamblinginContests} is very
simple, the solution is remarkably rich and subtle. Firstly, in equilibrium
agents must use randomised strategies, so that the level at which
the agent should stop is stochastic. Secondly, the set of values at
which the agent should stop forms an interval which is bounded above.
Several variants are discussed in Seel \& Strack, including the extension
to the asymmetric case where the starting values of the processes
observed by the agents are different.

We will consider the following variant of the problem. The agent's
choice of stopping rule determines her stopped value. But if with
hindsight we look at the best possible time she could have chosen
then we get a maximum value she might have attained. We consider a
problem in which the agent receives a reward of 1 if her stopped value
is the highest across all agents, but she is penalised $K$ if her
stopped value is not highest, and if she had stopped at the maximum
value she might have attained then she would have been the winner.
If she is not the winner, and there is no strategy she might have
followed which would led to her being the winner then her score is
zero.

In fact we consider three variants of the problem, in which an omniscient
being (or the agent's supervisor) penalises the agent for stopping
too soon, for stopping too late, or for stopping too soon or too late.
In the first case the agent faces regret over stopping too soon, and
we consider the maximum value to be the maximum value attained by
her process after the moment she chose to stop. In the second case
the maximum is taken only over that part of the path which occurs
before the chosen stopping time and the agent faces regret over stopping
too late. In the third case we take the maximum over the whole path.

Our results are that in the first problem, the effect of the penalty
is precisely equivalent to an increase in the number of opponents.
An increase in $K$ incentives the agent to aim for higher values,
at the cost of stopping at low values more often. This is the same
as the effect of competition from more opposing agents. In the second
problem, which is both harder and more interesting, the optimal strategy
is modified in a more subtle way. This case is relevant if the agent's
process is unobservable from the point at which is is stopped, for
instance if it is the gains from trade process arising from a dynamic
investment strategy chosen by the agent. Now the agent faces a risk
of a penalty whenever she stops below the value of the current maximum.
For this reason she is reluctant to do so, although it is also sub-optimal
to wait until her process hits zero, as this is a sure losing strategy.
An increase in $K$ incentives her to stop more quickly. The third
problem might be expected to be a combination of the two previous
problems, but in fact there is a natural simplification which leads
to the optimum being the same solution as the original Seel \& Strack
problem.

The remainder of this paper is constructed as follows. In the next
section we describe the original contest introduced in Seel \& Strack
\cite{GamblinginContests}. We rederive the Nash equilibrium, using
a different approach based on Lagrangian method. Then in Section \ref{sec:Regrets}
we introduce the problem with regret, which we then solve in the three
cases described above in Sections \ref{sec:Future-regrets}, \ref{sec:Past-regrets}
and \ref{sec:All-regrets}. Finally in Section \ref{sec:Derivation}
we explain the origin of the optimal multipliers and the candidate
Nash equilibrium distribution. The ideas of the proofs are different
to those in \cite{GamblinginContests} in that instead of trying to
write down the value function for the agents we use a Lagrangian sufficiency
theorem. This brings new insights and yields a simpler proof even
in the standard problem of Seel \& Strack and facilitates our analysis
in the extended problem.

\section{Contest without regret\label{sec:Without-regrets}}

\subsection{The model\label{sub:The-model}}

There are $n$ players with labels $i\in I=\{1,2,\ldots,n\}$ who
take part in the contest. Player $i$ privately observes the continuous-time
realisation of a Brownian motion $X^{i}=(X_{t}^{i})_{t\in\mathbb{R}^{+}}$
absorbed at zero with $X_{0}^{i}=x_{0}$ where $x_{0}$ is independent
of $i$ and positive. Let $\mathcal{F}_{t}^{i}=\sigma(\{X_{s}^{i}:s<t\})$
and set $\mathbb{F}^{i}=(\mathcal{F}_{t}^{i})_{t\geq0}$.

The space of strategies for agent $i$ is the space of $\mathbb{F}^{i}$-stopping
times $\tau^{i}$. Since zero is absorbing for $X^{i}$, without loss
of generality we may restrict attention to $\tau^{i}\leq H_{0}^{i}=\inf\{t\geq0:X_{t}^{i}=0\}$.
Player $i$ observes her own process $X^{i}$, but not $X^{j}$
for $j\neq i$; nor does she observe the stopping times chosen by the
other agents. Moreover, the processes $X^{i}$ are independent.

The player who stops at the highest value wins unit reward, that is,
$\forall i\in I$, player $i$ wins $1$ if she stops at time $\tau^{i}$
such that $X_{\tau^{i}}^{i}>X_{\tau^{j}}^{j}$ $\forall j\neq i$.
If there are $k$ players who stop at the equal highest value then
these players each win $\frac{1}{k}$. Therefore player $i$ with
stopping value $X_{\tau^{i}}^{i}$ receives payoff 
\[
\frac{1}{k}\mathbf{1}_{\{X_{\tau^{i}}^{i}=\max_{j\in I}X_{\tau^{j}}^{j}\}},
\]
where $k=\left|\left\{ i\in I:X_{\tau^{i}}^{i}=\max_{j\in I}X_{\tau^{j}}^{j}\right\} \right|$.

The key insight of Seel \& Strack~\cite{GamblinginContests} is to
observe that the problem of choosing the optimal stopping time can
be reduced to a problem of finding the optimal law for $X_{\tau}$
or equivalently an optimal target distribution. The payoffs to the
agents only depend upon $\tau^{i}$ via the distribution of $X_{\tau^{i}}^{i}$.
Hence, the problem can be considered in two stages, firstly find an
optimal target distribution $F^{i}$, and then verify that there is
a choice of $\tau^{i}$ such that $X_{\tau^{i}}^{i}$ has law $F^{i}$.
But, the problem of finding $\tau$ such that $X_{\tau}$ has law
$F$ is a classical problem in probability theory, and is known as
the Skorokhod embedding problem (Skorokhod~\cite{Skorokhod:65}).
Since $X$ is a Brownian motion started at $x_{0}$ and absorbed at
$0$, any distribution on $\mathbb{R}^{+}$ with mean less than or
equal to $x_{0}$ can be embedded with a finite stopping time $\tau$
(and conversely, for any $\tau$ the law of $X_{\tau}$ has mean less
than or equal to $x_{0}$).

Note that there are multiple solutions to the Skorokhod embedding
problem for $F$, so there will be several alternative stopping rules
which will bring equal probability of success to an agent. However,
we expect that the optimal target distribution is unique.

Our aim is to find Nash equilibria for the problem. By the above remarks,
a Nash equilibria can be identified with a family of distribution
functions $(F^{i})_{i\in I}$. We will say that $(F^{i})_{i\in I}$
is a Nash equilibrium if, for each $i\in I$, if the other agents
use stopping rules $\tau^{j}$ such that $X_{\tau^{j}}^{j}\sim F^{j}$,
then the optimal target distribution for agent $i$ is $F^{i}$, and
she may use any stopping rule $\tau^{i}$ such that $X_{\tau^{i}}^{i}\sim F^{i}$.
We will say a Nash equilibrium is symmetric if $F^{i}$ does not depend
on $i$, and we will say that a Nash equilibrium is atom-free if each
$F^{i}$ is atom-free. Given the symmetry of the situation in the
sense that each agent observes a martingale process started from the
same level $x_{0}$, it seems natural that a Nash equilibria is symmetric.
Moreover, simple arguments over rearranging mass can be used to show
that it is never optimal for two agents to put mass at the same positive
point $x$ --- either of them could benefit by modifying the target
distribution to put a proportion $N/(N+1)$ of this mass at $(x+N^{-2})$
and a proportion $1/(N+1)$ at $(x-N^{-1})$ -- and then it is possible
to deduce that any optimal solution is atom-free.

\begin{thm}\label{thm:Seel}{[}Seel \& Strack~\cite{GamblinginContests}{]}
Any Nash equilibrium has the property that it is symmetric and atom-free.
\end{thm}

\begin{remark} The fact that the Nash equilibria is atom-free relies
on the fact that the situation is symmetric in the sense that all
agents stop Brownian motions started from a common value $x_{0}$.
If the agents observe processes with different starting points, then
the Nash equilibria may have masses at zero for some agents. In that
case, for a Nash equilibrium, no agent places mass at a positive point,
and at least one agent has an atom-free distribution. We will only
consider the symmetric case. \end{remark}

\begin{remark} In the Seel \& Strack setting, which we call the standard
case, we will focus on proving that there exists a unique symmetric,
atom-free Nash equilibria. It will then follow from the results of
Seel \& Strack~\cite{GamblinginContests} that this is the unique
Nash equilibrium for our problem. Our methods can be extended to show
that every Nash equilibrium has the property that it is symmetric
and atom-free, but we will not present those arguments here. When
we consider the problem with a penalty for using a losing strategy
when a winning strategy exists, we will again prove the existence
of a unique symmetric atom-free Nash equilibrium, but intuition gained
from the standard case indicates that this equilibrium is unique.
\end{remark}

\subsection{Derivation of the equilibrium distribution\label{sub:Derivation-no-regret}}

This section is devoted to a proof of the following result, first
proved in Seel and Strack~\cite{GamblinginContests} using different
methods.

\begin{thm} \label{thm:no-regret} There exists a symmetric, atom-free
Nash equilibrium for the problem for which $X_{\tau^{i}}^{i}$ has
law $F(x)$, where for $x\geq0$ 
\[
F(x)=\min\left\{ \sqrt[n-1]{\frac{x}{nx_{0}}},1\right\} 
\]
\end{thm}

\begin{proof} Let $\mathcal{A}$ be the set of non-decreasing functions
$f:[0,\infty)\mapsto[0,\infty)$ which are null at zero, and let $\mathcal{A}_{D}(x_{0})$
be the subset of $\mathcal{A}$ corresponding to distribution functions
of random variables with mean $x_{0}$. Then, 
\[
\mathcal{A}_{D}(x_{0})=\left\{ f:[0,\infty)\mapsto[0,\infty)\;\textrm{such that\;}f(0)=0,f\;\mbox{non-decreasing},\;\lim_{x\uparrow\infty}f(x)=1,\int_{0}^{\infty}xf(dx)=x_{0}\right\} .
\]

We seek a symmetric atom-free Nash equilibrium. Since there are no
atoms, we do not need to consider how to break ties and a symmetric
Nash equilibrium is identified with a distribution function $G^{*}\in\mathcal{A}_{D}(x_{0})$
with the property that 
\[
\int_{0}^{\infty}G^{*}(x)^{n-1}G^{*}(dx)\geq\int_{0}^{\infty}G^{*}(x)^{n-1}G(dx)\hspace{10mm}\forall G\in\mathcal{A}_{D}(x_{0}).
\]

Suppose that the other players all choose $F(x)$ as their target
distribution. Then the problem facing the agent is to choose $G$
to solve 
\begin{equation}
\max_{G(x)\in\mathcal{A}}\int_{0}^{\infty}F(x)^{n-1}G(dx)\mbox{ subject to }\int_{0}^{\infty}xG(dx)=x_{0}\mbox{ and }\int_{0}^{\infty}G(dx)=1.\label{eq:Optimizationproblemnoregretssymmetric}
\end{equation}
Introducing multipliers $\lambda$ and $\gamma$ for the two constraints,
the Lagrangian for the optimization problem (\ref{eq:Optimizationproblemnoregretssymmetric})
is then 
\[
\mathcal{L}_{F}(G;\lambda,\gamma)=\int_{0}^{\infty}\left[F(x)^{n-1}-\lambda x-\gamma\right]G(dx)+\lambda x_{0}+\gamma.
\]
Now we state a variant of the Lagrangian sufficiency theorem for our
problem.

\begin{prop}\label{prop:No-regret} If $G^{*}$, $\lambda^{*}$ and
$\gamma^{*}$ exist such that $G^{*}\in\mathcal{A}_{D}(x_{0})$ and
\begin{equation}
\mathcal{L}_{G^{*}}(G^{*};\lambda^{*},\gamma^{*})\geq\mathcal{L}_{G^{*}}(G;\lambda^{*},\gamma^{*})\label{eq:LST}
\end{equation}
for all $G\in\mathcal{A}$, then $G^{*}$ is a symmetric, atom-free
Nash equilibrium. \end{prop}

\begin{proof} If $G\in\mathcal{A}_{D}(x_{0})$ then 
\[
\int_{0}^{\infty}G^{*}(x)^{n-1}G(dx)=\mathcal{L}_{G^{*}}(G;\lambda^{*},\gamma^{*}).
\]
Then, under the hypotheses of the proposition, 
\[
\int_{0}^{\infty}G^{*}(x)^{n-1}G^{*}(dx)=\mathcal{L}_{G^{*}}(G^{*};\lambda^{*},\gamma^{*})\geq\mathcal{L}_{G^{*}}(G;\lambda^{*},\gamma^{*})=\int_{0}^{\infty}G^{*}(x)^{n-1}G(dx).
\]
\end{proof}

Return to the proof of Theorem \ref{thm:no-regret}. On $[0,\infty)$
let $G^{*}(x)=\min\left\{ \sqrt[n-1]{x/(nx_{0})},1\right\} $, $\lambda^{*}=1/(nx_{0})$
and $\gamma^{*}=0$. We verify that for these multipliers (\ref{eq:LST})
holds and that $G^{*}\in\mathcal{A}_{D}(x_{0})$. The latter follows
immediately from the explicit form of $G^{*}$. For the former 
\begin{align*}
\mathcal{L}_{G^{*}}(G;\lambda^{*},\gamma^{*}) & =\int_{0}^{\infty}\left[G^{*}(x)^{n-1}-\lambda^{*}x-\gamma^{*}\right]G(dx)+\lambda^{*}x_{0}+\gamma^{*}\\
 & =\int_{nx_{0}}^{\infty}\left[1-\frac{x}{nx_{0}}\right]G(dx)+\frac{1}{n}\leq\frac{1}{n}=\mathcal{L}_{G^{*}}(G^{*};\lambda^{*},\gamma^{*}).
\end{align*}
Thus there exists a symmetric, atom-free Nash equilibrium of the given
form. \end{proof}

\begin{remark} Seel \& Strack~\cite{GamblinginContests} solve the
problem by writing down a candidate value function for the problem,
and then verifying that the candidate value function is a martingale
under an optimal stopping rule for each agent. \end{remark}

\begin{remark}

Our results can be extended to the case where the processes observed
by the agents are independent copies of some time-homogeneous diffusion
process $Y$ which converges almost surely to the lower bound on its
state space. The idea is to use a change of scale, and, in the setting
of Skorokhod embeddings, can be traced back to Az\'{e}ma and Yor~\cite{AzemaYor79b}.
In addition to Brownian motion (absorbed at zero), canonical examples
include exponential Brownian motion and drifting Brownian motion with
negative drift (with or without absorption at zero). Seel \& Strack~\cite{GamblinginContests}
consider the problem when $Y^{i}$ is a Brownian motion with positive
drift, absorbed at zero, but then they need to impose a further condition
on the model parameters to ensure the finiteness of the candidate
solution.

Let the state space of the time-homogeneous diffusion $Y$ be an interval
$S$ with endpoints $\{l,r\}$ with $-\infty\leq l<Y_{0}=y_{0}<r\leq\infty$.
Suppose that $Y$ is a solution of the stochastic differential equation
$dY=a(Y)dB+b(Y)dt$ and let $s=s(y)$ be the scale function%
\footnote{If $Y$ is exponential Brownian motion, $dY=aYdB+bYdt$ then $s(y)=y^{\kappa}$
with $\kappa=1-2b/a^{2}$. Note that we need parameters such that
$\kappa>0$ to ensure that $Y$ is transient to zero. If $Y$ is downward
drifting Brownian motion, $dY=adB+bdt$ with $b<0$, then $s(y)=e^{-\eta y}$
with $\eta=2b/a^{2}$.%
}. Then $s$ is an increasing solution of $a(y)^{2}s''(y)+2b(y)s'(y)=0$
and $X=s(Y)$ is a continuous local martingale with starting value
$x_{0}=s(y_{0})$, and hence a time-change of Brownian motion. Moreover,
our assumption that $Y$ converges to the lower boundary implies that
$s(l)$ is finite whereas $s(r)=\infty$ and without loss of generality
we may set $s(l)=0$. Then $X=s(Y)$ converges to zero almost surely
(and if zero can be reached in finite time, then zero is absorbing).

Note that $s(\cdot)$ is a continuous strictly increasing function.
Hence the payoff of player $i$ with stopping value $Y_{\tau^{i}}^{i}$
can be expressed as $\frac{1}{k}\mathbf{1}_{\{X_{\tau^{i}}^{i}=\max_{j\in I}X_{\tau^{j}}^{j}\}},$
where $k=\left|\left\{ i\in I:X_{\tau^{i}}^{i}=\max_{j\in I}X_{\tau^{j}}^{j}\right\} \right|$,
and $X_{\tau^{i}}^{i}=s(Y_{\tau^{i}}^{i})$. Then the contest in which
players privately observe $Y^{i}$ is equivalent to the contest in
which players privately observe $X^{i}$, and the choice of the optimal
$\tau^{i}$ is the same for both problems. In particular, if we have
a Nash equilibrium for which $\tau^{i}$ is optimal for the processes
$X^{i}$, then we also have a Nash equilibrium for the processes $Y^{i}$.

The problem is then to find a Nash equilibrium $(G^{i})_{i\in I}$
for $Y_{\tau^{i}}^{i}$ and then verify that there exists $\tau^{i}$
such that $Y_{\tau^{i}}^{i}$ has law $G^{i}$. Under our transformation,
this is the same as finding a Nash equilibrium $(F^{i})_{i\in I}$
for $X_{\tau^{i}}^{i}$ where $F^{i}=G^{i}\circ s^{-1}$, where $s^{-1}$
is the inverse of $s$. To solve the problem for $X$, then either
we argue that the only properties of $X$ that we use are the strong
Markov property, the local martingale property, and the fact that
$X$ converges to zero, so that the theory of this section applies
to the local martingale diffusion $X$, or we argue that since $X$
is a non-negative martingale diffusion, $X$ is a time-change of Brownian
motion and $X_{t}=B_{\Gamma_{t}}$ for some increasing functional
$\Gamma_{t}$. Then if $F$ is any distribution with mean less than
or equal to $x_{0}$, and $\sigma$ is a stopping time such that $B_{\sigma}\sim F$,
then we may take $\tau=\Gamma^{-1}\circ\sigma$ and then $X_{\tau}=B_{\sigma}\sim F$
and $\tau$ is an embedding of $G$ in $Y$.

Note that if $X_{\tau}\sim F$ and $F$ has mean $x_{1}<x_{0}$, then
there exists $(\tilde{F},\tilde{\tau})$ such that $\tilde{F}$ has
mean $x_{0}$, $\tilde{F}\geq F$ and $X_{\tilde{\tau}}\sim\tilde{F}$.
Clearly $\tilde{\tau}$ dominates $\tau$ as a strategy. Hence we
may restrict attention to stopping times $\tau$ such that the distribution
$F$ of $X$ has mean $x_{0}$, and then $(X_{t\wedge\tau})_{t\geq0}$
is a martingale and not just a local martingale. \end{remark}

\section{Contests with regret\label{sec:Regrets}}

Our goal is to solve an extended version of the problem in which agents
are penalised for following losing strategies, if they had an alternative
stopping rule which would have won the contest. The idea is that there
is an omniscient judge who can observe the path of $X^{i}$, and not
just the stopped value, and who penalises the agent for the failure
to use a winning stopping rule if such a strategy exists. This judge
represents the supervisor of the agent, and the agent faces penalties
(such as dismissal) in cases where after the fact she is seen to have
followed a losing strategy, when a winning strategy existed.

In our model there are $n$ contestants, each of whom privately observes
their own process $X^{i}$ and the player with the highest stopping
value wins unit reward. That is, $\forall i\in I$, player $i$ wins
$1$ if she stops at time $\tau^{i}$ such that $X_{\tau^{i}}^{i}>X_{\tau^{j}}^{j}$
$\forall j\neq i$. In addition the player is penalised $K\geq0$
if her stopped value is not highest, and if she had an alternative
strategy which would, with the benefit of hindsight, have allowed
her to win. (The case $K=0$ corresponds to the standard problem.)
Given that the best strategy for agent $i$ is to stop at the maximum
value $M^{i}$ attained by $X^{i}$ this means that player $i$ loses
$K$ if she stops at $\tau^{i}$ such that $X_{\tau^{i}}^{i}<\max_{j\neq i}X_{\tau^{j}}^{j}<M^{i}$.

As before, ties are broken randomly. If there are $k$ players who
stop at the highest value then these players each wins $\frac{1}{k}$.
Further, player $i$ loses $K_{2}$ if she stops at $\tau^{i}$ such
that $X_{\tau^{i}}^{i}<\max_{j\neq i}X_{\tau^{j}}^{j}=M^{i}$, where
$0\leq K_{2}\leq K$. Hence player $i$ who stops at $X_{\tau^{i}}^{i}$
with maximum value $M^{i}$ has payoff 
\[
\frac{1}{k}\mathbf{1}_{\{X_{\tau^{i}}^{i}=\max_{j\in I}X_{\tau^{j}}^{j}\}}-K\mathbf{1}_{\{X_{\tau^{i}}^{i}<\max_{j\neq i}X_{\tau^{j}}^{j}<M^{i}\}}-K_{2}\mathbf{1}_{\{X_{\tau^{i}}^{i}<\max_{j\neq i}X_{\tau^{j}}^{j}=M^{i}\}},
\]
where $k=\left|\left\{ i\in I:X_{\tau^{i}}^{i}=\max_{j\in I}X_{\tau^{j}}^{j}\right\} \right|$.

Our objective is to find a Nash equilibria which is represented by
a family of stopping rules $(\tau^{i})$. Since the values $(X_{\tau^{i}}^{i},M^{i})$
are a sufficient statistic for the problem, the Nash equilibria can
be characterised by the law $\nu^{i}$ of $(X_{\tau^{i}}^{i},M^{i})$.
Then in equilibrium, the agent can use any stopping rule for which
$(X_{\tau^{i}}^{i},M^{i})$ has law $\nu^{i}$. We write $F^{i}$
for the marginal of $\nu$ which corresponds to the law of $X_{\tau^{i}}^{i}$.

In the standard case, every Nash equilibrium is symmetric and atom-free.
In our generalised setting we will limit our search to symmetric atom-free
Nash equilibria, and we will show that there exists a unique such
equilibrium. Motivated by the situation in the standard case we conjecture
that every Nash equilibria is symmetric and atom-free and therefore
that we have found the unique equilibrium.

\begin{remark}Since there are no atoms, the probability of a tie
is zero. Thus neither the method of breaking ties nor the value of
$K_{2}$ will affect our results. \end{remark}

Suppose that the other players all choose $F(x)$ as their target
distribution of $X_{\tau}$. Then agent $i$ aims to choose a feasible
measure $\nu(x,y)$ for $(X_{\tau^{i}}^{i},M^{i})$ to maximise 
\begin{align}
\mathbb{E}\left[F(X_{\tau^{i}}^{i})^{n-1}\right] & -K\mathbb{E}\left[F(M^{i})^{n-1}-F(X_{\tau^{i}}^{i})^{n-1}\right]\nonumber \\
 & =(1+K)\mathbb{E}\left[F(X_{\tau^{i}}^{i})^{n-1}\right]-K\mathbb{E}\left[F(M^{i})^{n-1}\right],\nonumber \\
 & =\int_{0}^{\infty}\int_{0}^{\infty}\left[(1+K)F(x)^{n-1}-KF(y)^{n-1}\right]\nu(dx,dy),\label{eq:expected payoff}
\end{align}
which is her expected payoff when she stops at $\tau^{i}$.

So far we have been imprecise about the definition of $M^{i}$. The
quantity $M^{i}$ represents the maximum the agent could have achieved.
Depending on the interpretation, this could be the maximum over the
entire path $M^{i}=\max\{X_{t}^{i};0\leq t\leq H_{0}^{i}\}$, or it
could be that only that part of the path before the agent's chosen
stopping time is considered, $M^{i}=\max\{X_{t}^{i};0\leq t\leq\tau^{i}\}$,
or only that part of the path after the agent's chosen stopping time,
$M^{i}=\max\{X_{t}^{i};\tau^{i}\leq t\leq H_{0}^{i}\}$, where $H_{0}^{i}=\inf\{t\in\mathbb{R}^{+}:X_{t}^{i}=0\}$.
These different interpretations will lead to different Nash equilibria.
We consider the three cases separately in the next three sections.

\section{Contest with regret over future failure to stop\label{sec:Future-regrets}}

In this section we consider the contest in which the agent is penalised
for stopping too soon. We consider the maximum value $M^{i}$ to be
defined by 
\[
M^{i}:=M_{[\tau^{i},H_{0}^{i}]}^{i}=\sup_{\tau^{i}\leq t\leq H_{0}^{i}}X_{t}^{i}.
\]

\begin{thm}\label{thm:future-regrets}There exists a symmetric, atom-free
Nash equilibrium for the problem for which $X_{\tau^{i}}^{i}$ has
law $F(x)$, where for $x\geq0$ 
\[
F(x)=\min\left\{ \sqrt[N-1]{\frac{x}{Nx_{0}}},1\right\} 
\]
with $N=n+K(n-1)$. \end{thm}

\begin{remark} The agent follows exactly the same optimal strategy
as an agent in a different setup, where there is no penalty, but the
total number of contestants is increased to $N=n+K(n-1)$. \end{remark}

\begin{proof}Denote by $\nu$ the joint distribution of $X_{\tau^{i}}^{i}$
and $M_{[\tau^{i},H_{0}^{i}]}^{i}$ and denote by $G(x)$ the marginal
distribution of $X_{\tau^{i}}^{i}$. Then using the strong Markov
property and the martingale property of $X$, 
\begin{align}
\nu([0,x]\times[0,y]) & =\mathbb{P}(X_{\tau^{i}}^{i}\leq x,M_{[\tau^{i},H_{0}^{i}]}^{i}\leq y)=\int_{0}^{x}\mathbb{P}(M_{[\tau^{i},H_{0}^{i}]}^{i}\leq y|X_{\tau^{i}}^{i}=z)G(dz)\nonumber \\
 & =\int_{0}^{x}\mathbb{P}(H_{0}^{i}<H_{y}^{i}|X_{0}^{i}=z)G(dz)=\int_{0}^{x}\frac{y-z}{y}G(dz),\label{eq:joint distr future failure}
\end{align}
where $H_{y}^{i}=\inf\{t\in\mathbb{R}^{+}:X_{t}^{i}=y\}$.

Suppose that the other players all choose $F(x)$ as their target
distribution of $X_{\tau}$. Substituting (\ref{eq:joint distr future failure})
into (\ref{eq:expected payoff}), the expected payoff of player $i$
becomes 
\begin{align*}
\int_{0}^{\infty}\int_{0}^{\infty} & \left[(1+K)F(x)^{n-1}-KF(y)^{n-1}\right]\nu(dx,dy)\\
= & \int_{0}^{\infty}(1+K)F(x)^{n-1}G(dx)-\int_{0}^{\infty}\int_{x}^{\infty}KF(y)^{n-1}\frac{x}{y^{2}}dyG(dx)\\
= & \int_{0}^{\infty}\left[(1+K)F(x)^{n-1}-Kx\int_{x}^{\infty}\frac{F(y)^{n-1}}{y^{2}}dy\right]G(dx).
\end{align*}

Given other players' choices, player $i$ would like to choose $G\in\mathcal{A}$
to solve 
\begin{equation}
\max_{G\in\mathcal{A}}\int_{0}^{\infty}\left[(1+K)F(x)^{n-1}-Kx\int_{x}^{\infty}\frac{F(y)^{n-1}}{y^{2}}dy\right]G(dx)\mbox{ subject to }\int_{0}^{\infty}xG(dx)=x_{0}\mbox{ and }\int_{0}^{\infty}G(dx)=1.\label{eq:Optimization problem future failure symmetric}
\end{equation}
Introducing multipliers $\lambda$ and $\gamma$ for the two constraints,
the Lagrangian for the optimization problem (\ref{eq:Optimization problem future failure symmetric})
is then 
\[
\mathcal{L}_{F}(G;\lambda,\gamma)=\int_{0}^{\infty}\left[(1+K)F(x)^{n-1}-Kx\int_{x}^{\infty}\frac{F(y)^{n-1}}{y^{2}}dy-\lambda x-\gamma\right]G(dx)+\lambda x_{0}+\gamma.
\]
On $[0,\infty)$ let $G^{*}(x)=\min\left\{ 1,\sqrt[N-1]{x/(Nx_{0})}\right\} $,
$\lambda^{*}=1/(Nx_{0})$ and $\gamma^{*}=0$, where $N=n+K(n-1)$.
It is easy to check that $G^{*}\in\mathcal{A}_{D}(x_{0})$. Moreover,
\begin{align*}
\mathcal{L}_{G^{*}}(G;\lambda^{*},\gamma^{*}) & =\int_{0}^{\infty}\left[(1+K)G^{*}(x)^{n-1}-Kx\int_{x}^{\infty}\frac{G^{*}(y)^{n-1}}{y^{2}}dy-\lambda^{*}x-\gamma^{*}\right]G(dx)+\lambda^{*}x_{0}+\gamma^{*}\\
 & =\int_{Nx_{0}}^{\infty}\left[1-\frac{x}{Nx_{0}}\right]G(dx)+\frac{1}{N}\leq\frac{1}{N}=\mathcal{L}_{G^{*}}(G^{*};\lambda^{*},\gamma^{*}).
\end{align*}
Hence, by the Lagrangian sufficiency theorem (Proposition \ref{prop:No-regret})
$G^{*}$ is a symmetric, atom-free Nash equilibrium. \end{proof}

\begin{remark} In this version of the problem, the stopping decision
depends on the current value of $X$ alone, and not on the current
maximum. This is because the penalty depends on the future maximum,
which conditional on the current value of the process is independent
of the past maximum. \end{remark}

\section{Contest with regret over past failure to stop\label{sec:Past-regrets}}

This section discusses the contest with regret over past failure to
stop, that is player is penalised when she could have won if she had
stopped sooner. This case is relevant when the omniscient being can
only observe the realisation of $X^{i}$ up to the stopping time chosen
by the agent. In this case the maximum value $M^{i}$ is defined by
\[
M^{i}:=M_{\tau^{i}}^{i}=\sup_{0\leq t\leq\tau^{i}}X_{t}^{i}.
\]

Consider the problem facing a single agent under the assumption that
the strategies of the competing agents are fixed. Temporarily we drop
the subscript denoting the label of the agent. Recall that the payoff
to the agent is $(1+K)F(X_{\tau})^{n-1}-KF(M_{\tau})^{n-1}$. For
a continuous martingale Kertz and R\"{o}sler~\cite{KertzRosler} characterise
all possible joint laws of $(X_{\tau},M_{\tau})$ and hence the problem
is reduced to a search over measures with these characteristics. However,
an alternative is to split the optimisation problem into a two-stage
procedure: first for any feasible distribution of $X_{\tau}$ (a non-negative
random variable with mean $x_{0}$) find the joint law of $(X_{\tau},M_{\tau})$
for which $M_{\tau}$ is as small as possible in distribution ---
such a joint law exists by results of Perkins~\cite{Perkinssolution}
--- and then minimise a modified objective function over feasible
laws of $X_{\tau}$.

For a given law of $X_{\tau}$ the joint law of $(X_{\tau},M_{\tau})$
for which $M_{\tau}$ is minimised is such that mass is placed only
on the set $A=\{(x,x);x\geq x_{0}\}\cup\{(x,\Phi(x));x<x_{0}\}$ where
$\Phi:(0,x_{0})\mapsto(x_{0},\infty)$ is a decreasing function (and
if $X_{\tau}$ is atom-free, a strictly decreasing function). Let
$\phi$ be inverse to $\Phi$. Then, if $G$ denotes the marginal
law of $X$, we can conclude from Doob's submartingale inequality%
\footnote{By the final remark of Section \ref{sub:Derivation-no-regret}, we
may assume $(X_{t\wedge\tau})_{t\geq0}$ is a martingale, and then
we have, for $m\geq x_{0}$, $m\mathbb{P}(M_{\tau}\geq m)=\mathbb{E}[X_{\tau};M_{\tau}\geq m]$.
Hence $0=\mathbb{E}[m-X_{\tau};M_{\tau}\geq m]$. %
}, in conjunction with the set identity $(M_{\tau}\geq m)=(X_{\tau}\geq m)\cup(X_{\tau}\leq\phi(m))$,
that for $m\geq x_{0}$ 
\begin{equation}
0=\mathbb{E}[m-X_{\tau};X_{\tau}\geq m]+\mathbb{E}[m-X_{\tau};X_{\tau}\leq\phi(m)]=\int_{m}^{\infty}(m-y)G(dy)+\int_{0}^{\phi(m)}(m-y)G(dy)\label{eqn:PerkinsHPdefnphiI}
\end{equation}
which, since $X_{\tau}$ has mean $x_{0}$, is equivalent to 
\begin{equation}
0=m-x_{0}+(m-\phi(m))G(\phi(m))-\int_{\phi(m)}^{m}G(y)dy.\label{eq:Gphi}
\end{equation}
In differential form, assuming $G$ and $\phi$ are differentiable,
this becomes 
\begin{equation}
0=\phi'(m)(m-\phi(m))G'(\phi(m))+1+G(\phi(m))-G(m).\label{eqn:PerkinsHPdefnphiD}
\end{equation}
It follows from the results of Perkins~\cite{Perkinssolution} and
Hobson and Pedersen~\cite{HobsonPedersen}, that if $G$ is the law
of an atom-free non-negative random variable, then there exists a
decreasing function $\phi$ solving (\ref{eqn:PerkinsHPdefnphiI}).
Further, if $\xi$ is a random variable such that for $s\geq x_{0}$
\[
\mathbb{P}(\xi\geq s)=\exp\left(-\int_{(x_{0},s)}\frac{G(du)}{1-G(u)+G(\phi(u))}\right)
\]
and if $\tau=\tau_{\xi}\wedge\tau_{\phi}$ where $\tau_{\xi}=\inf\{t>0|M_{t}\geq\xi\}$
and $\tau_{\phi}=\inf\{t>0|X_{t}\leq\phi(M_{t})\}$, then $X_{\tau}$
has law $G$ and $(X_{\tau},M_{\tau})$ places no mass off $A$. Moreover,
amongst the class of joint laws for $(X_{\tau},M_{\tau})$ such that
$X_{\tau}$ has law $G$, $M_{\tau}$ is as small as possible in distribution.

\begin{thm}\label{thm:Past-regret} Suppose there exists a finite
real number $r>x_{0}$, a once differentiable strictly decreasing
function $\phi:[x_{0},r]\mapsto[0,x_{0}]$, a thrice differentiable
strictly increasing and strictly convex function $\psi:[x_{0},r]\mapsto[0,1]$
and a once differentiable strictly decreasing function $\theta:[x_{0},r]\mapsto[0,1]$
such that $\phi$, $\psi$ and $\theta$ solve the following system
of equations \begin{numcases}{(\dagger)\quad}
\phi'(y)\psi'(y)=(1+K)\theta'(y),\label{eq:thm-eq1}\\
K\psi'(y)=(y-\phi(y))\psi''(y), \label{eq:thm-eq2}\\
\frac{y-\phi(y)}{n-1}\theta'(y)=\left(\psi(y)^{\frac{1}{n-1}}-1\right)\theta(y)^{\frac{n-2}{n-1}}-\theta(y),\label{eq:thm-eq3}
\end{numcases}and satisfy that $\phi(x_{0})=x_{0}$, $\psi(r)=1$, $\psi'(r-)=\frac{K+1}{r}$,
$\psi''(r-)=\frac{K(K+1)}{r^{2}}$ and $\theta(x_{0})=\psi(x_{0})$.

i) Then 
\begin{equation}
\theta(y)=\psi(y)-\frac{K}{K+1}\frac{\psi'(y)^{2}}{\psi''(y)}.\label{eq:thm-theta}
\end{equation}

ii) Moreover, there exists a symmetric, atom-free Nash equilibrium
for the problem for which $X_{\tau^{i}}^{i}$ and $M_{\tau^{i}}^{i}$
have joint law $\nu^{*}$ that is determined by the marginal distribution
$G^{*}$ of $X_{\tau^{i}}^{i}$, given by $G^{*}(x)=0$ for $x\leq0$,
$G^{*}(x)=1$ for $x\geq r$ and 
\begin{equation}
G^{*}(x)=\begin{cases}
\theta(\phi^{-1}(x))^{\frac{1}{n-1}} & ,\textrm{ if }0<x<x_{0},\\
\psi(x)^{\frac{1}{n-1}} & ,\textrm{ if }x_{0}\leq x<r,
\end{cases}\label{eq:thm-F-past-failure}
\end{equation}
otherwise, and the conditional distribution of $M_{\tau^{i}}^{i}$
given $X_{\tau^{i}}^{i}$ such that 
\[
M_{\tau^{i}}^{i}=\begin{cases}
X_{\tau^{i}}^{i} & ,\textrm{ if }X_{\tau^{i}}^{i}\geq x_{0},\\
\phi^{-1}(X_{\tau^{i}}^{i}) & ,\textrm{ if }0\leq X_{\tau^{i}}^{i}<x_{0}.
\end{cases}
\]
\end{thm}

\begin{proof} The conditions in the theorem imply some properties
of function $\phi$: let $y=r-$ in (\ref{eq:thm-eq2}) then since
$\psi'(r-)=\frac{K+1}{r}$ and $\psi''(r-)=\frac{K(K+1)}{r^{2}}$
we have $\phi(r)=0$; since (\ref{eq:thm-eq2}) holds and by the positivity
of $\psi'$ and $\psi''$ we have $\phi(y)<y$ on $(x_{0},r)$.

i) Integrating (\ref{eq:thm-eq1}) with respect to $y$, 
\begin{equation}
(K+1)\theta(z)-(K+1)\theta(x_{0})=\int_{x_{0}}^{z}\phi'(y)\psi'(y)dy=\phi(z)\psi'(z)-\phi(x_{0})\psi'(x_{0}+)-\int_{x_{0}}^{z}\phi(y)\psi''(y)dy.\label{eq:7}
\end{equation}
Rearranging (\ref{eq:thm-eq2}) and integrating, 
\begin{align*}
\int_{x_{0}}^{z}\phi(y)\psi''(y)dy & =\int_{x_{0}}^{z}y\psi''(y)dy-\int_{x_{0}}^{z}K\psi'(y)dy=z\psi'(z)-x_{0}\psi'(x_{0}+)-(1+K)\psi(z)+(1+K)\psi(x_{0}).
\end{align*}
Combining this equation with (\ref{eq:7}) we find 
\begin{equation}
\theta(y)=\frac{1}{K+1}(\phi(y)-y)\psi'(y)+\psi(y).\label{eq:10}
\end{equation}
Then substituting (\ref{eq:thm-eq2}) into (\ref{eq:10}), (\ref{eq:thm-theta})
follows.

ii) Let $\mathcal{E}(x_{0})$ be the set of measures $\nu(dx,dy)$
on $[0,\infty)\times[0,\infty)$ such that $\nu(dx,dy)$ has no mass
on $\{(x,y):y<x\textrm{ or }y<x_{0}\}$, and let $\mathcal{E}_{D}(x_{0})$
be the subset of $\mathcal{E}(x_{0})$ corresponding to probability
measures of a pair of random variables $X\leq Y$ such that $X$ is
a continuous random variable with mean $x_{0}$ and $\mathbb{E}[X-z;Y\geq z]=0$
for all $z\geq x_{0}$. Note that the last equation comes from the
Doob's submartingale inequality, applied in the continuous martingale
case.

Suppose that the other players all choose $F(x)$ as the target distribution
of $X_{\tau}$. Then the aim of player $i$ is to choose $\nu\in\mathcal{E}(x_{0})$
to solve 
\begin{equation}
\max_{\nu\in\mathcal{E}(x_{0})}\left\{ \int_{0}^{\infty}\int_{0}^{\infty}\left[(1+K)F(x)^{n-1}-KF(y)^{n-1}\right]\nu(dx,dy)\right\} \label{eq:Optimization problem past failure symmetric}
\end{equation}
subject to $\int_{0}^{\infty}\int_{0}^{\infty}x\nu(dx,dy)=x_{0}$,
$\int_{0}^{\infty}\int_{0}^{\infty}\nu(dx,dy)=1$ and $\int_{x=0}^{\infty}\int_{y=z}^{\infty}(x-z)\nu(dx,dy)=0$
$\forall z\geq x_{0}$.

Introduce multipliers $\lambda$ and $\gamma$ for the first two constraints,
and for each $z\geq x_{0}$ introduce a Lagrange multiplier $\eta\left(z\right)$
for the last constraint: the constraint becomes $\int_{0}^{\infty}\int_{0}^{\infty}\int_{z=x_{0}}^{y}\left\{ \eta(z)(x-z)dz\right\} \nu(dx,dy)=0.$
Then the Lagrangian for the optimization problem (\ref{eq:Optimization problem past failure symmetric})
is 
\begin{align}
\mathcal{L}_{F}(\nu;\lambda,\gamma,\eta)\label{eqn:LagrangianP}\\
=\int_{0}^{\infty} & \int_{0}^{\infty}\left[(1+K)F(x)^{n-1}-KF(y)^{n-1}-\lambda x-\gamma-\int_{x_{0}}^{y}\eta(z)(x-z)dz\right]\nu(dx,dy)+\lambda x_{0}+\gamma.\nonumber 
\end{align}
By a simple extension of the Lagrangian sufficiency theorem given
in Proposition \ref{prop:No-regret} we have that

\begin{prop}\label{prop:Past-regret} If $\nu^{*}$, $\lambda^{*}$,
$\gamma^{*}$ and $\eta^{*}$ exist such that $\nu^{*}\in\mathcal{E}_{D}(x_{0})$
and 
\begin{equation}
\mathcal{L}_{G^{*}}(\nu^{*};\lambda^{*},\gamma^{*},\eta^{*})\geq\mathcal{L}_{G^{*}}(\nu;\lambda^{*},\gamma^{*},\eta^{*})\label{eq:LST-past}
\end{equation}
for all $\nu\in\mathcal{E}(x_{0})$, where $G^{*}(x)=\nu^{*}([0,x]\times(0,\infty))$,
then $\nu^{*}$ is a symmetric, atom-free Nash equilibrium. \end{prop}

On $[0,\infty)\times[0,\infty)$ let $\nu^{*}$ be the joint law given
in the theorem and $G^{*}$ be its marginal distribution with respect
to $X_{\tau}$. In particular, $G^{*}$ is given by (\ref{eq:thm-F-past-failure}).
Let $\lambda^{*}=\psi'(x_{0}+)$, $\gamma^{*}=\psi(x_{0})-x_{0}\psi'(x_{0}+)$,
$\eta^{*}(y)=\psi''(y)$ for $x_{0}<y<r$, and $\eta^{*}(y)=0$ for
$y\geq r$. We will show that for these multipliers (\ref{eq:LST-past})
holds and that $\nu^{*}\in\mathcal{E}_{D}(x_{0})$.

To prove $\nu^{*}\in\mathcal{E}_{D}(x_{0})$, we need to show $G^{*}(0)=0$,
$\lim_{y\uparrow\infty}G^{*}(y)=1$, $G^{*}(y)$ is continuous and
non-decreasing, $\int_{0}^{\infty}uG^{*}(du)=x_{0}$ and $\int_{x=0}^{\infty}\int_{y=z}^{\infty}(x-z)\nu^{*}(dx,dy)=0$
for all $z\geq x_{0}$.

Letting $y=r-$ in (\ref{eq:thm-theta}), we find $\theta(r)=0$.
Then since $\phi(r)=0$, we have $G^{*}(0)=\theta(\phi^{-1}(0))^{1/(n-1)}=\theta(r)^{1/(n-1)}=0$.
Moreover, $\lim_{y\uparrow\infty}G^{*}(y)=1$ follows from the finiteness
of $r$. We have $\phi^{-1}(x_{0})=x_{0}$ and hence $G^{*}$ is continuous
at $x_{0}$. Since both $\phi$ and $\theta$ are decreasing and continuous
on $[x_{0},r]$, $G^{*}\left(x\right)=\theta(\phi^{-1}(x))^{1/(n-1)}$
is increasing for $x\in[0,x_{0}]$. Then, since $\psi(y)$ is increasing
on $[x_{0},r]$, $G^{*}$ is continuous and non-decreasing on the
whole interval of $[0,r]$. Note that this implies $r=\sup\{x\geq0:G^{*}(x)<1\}$.

For $y>x_{0}$ we have $G^{*}(\phi(y))=\theta(y)^{1/(n-1)}$ and so
$\phi'(y)(G^{*})'(\phi(y))=\theta(y)^{1/(n-1)-1}\theta'(y)/(n-1)$.
Then, using (\ref{eq:thm-eq3}) 
\[
\phi'(y)(y-\phi(y))(G^{*})'(\phi(y))=\theta(y)^{\frac{2-n}{n-1}}\frac{y-\phi(y)}{n-1}\theta'(y)=\psi(y)^{\frac{1}{n-1}}-1-\theta(y)^{\frac{1}{n-1}}=G^{*}(y)-1-G^{*}(\phi(y)).
\]
Hence 
\[
\phi'(y)(y-\phi(y))(G^{*})'(\phi(y))+(1-\phi'(y))G^{*}(\phi(y))=G^{*}(y)-1-\phi'(y)G^{*}(\phi(y)),
\]
and integrating from $x$ to $r$ 
\begin{equation}
-(x-\phi(x))G^{*}(\phi(x))=-(r-x)+\int_{x}^{r}G^{*}(y)dy+\int_{0}^{\phi(x)}G^{*}(y)dy.\label{eq:Gstarphi}
\end{equation}
Then, setting $x=x_{0}$ we recover $x_{0}=\int_{0}^{r}(1-G^{*}(y))dy$
so that a random variable with distribution function $G^{*}$ has
mean $x_{0}$.

Finally, from its construction we have that $\nu^{*}$ only puts mass
on $A$. Hence, from (\ref{eqn:PerkinsHPdefnphiD}), 
\[
\int_{x=0}^{\infty}\int_{y=z}^{\infty}(x-z)\nu^{*}(dx,dy)=\int_{0}^{\phi(z)}(x-z)G^{*}(dx)+\int_{z}^{r}(x-z)G^{*}(dx)=0.
\]

Now we prove that (\ref{eq:LST-past}) holds. Let $L^{*}(x,y)=(1+K)G^{*}(x)^{n-1}-KG^{*}(y)^{n-1}-\lambda^{*}x-\gamma^{*}-\int_{x_{0}}^{y}\eta^{*}(z)(x-z)dz$
and then $\mathcal{L}_{G^{*}}(\nu;\lambda^{*},\gamma^{*},\eta^{*})=\int_{0}^{\infty}\int_{0}^{\infty}L^{*}(x,y)\nu(dx,dy)+\lambda^{*}x_{0}+\gamma^{*}.$

For notational convenience, extend the domain of $\psi$ to $[0,r]$
by defining $\psi(x)=\theta(\phi^{-1}(x))$ for $x\in[0,x_{0})$.
Then, for $x<x_{0}$, $\psi'(x)=\frac{\theta'(\phi^{-1}(x))}{\phi'(\phi^{-1}(x))}=\frac{\psi'(\phi^{-1}(x))}{K+1}>0$
where the last equality comes from (\ref{eq:thm-eq1}). Moreover $\psi''\left(x\right)=\frac{\psi''(\phi^{-1}(x))}{(1+K)\phi'(\phi^{-1}(x))}<0$.
Thus $\psi$ is increasing on $[0,r]$, $\psi''(x)<0$ if $x\in(0,x_{0})$
and $\psi''(x)>0$ if $x\in(x_{0},r)$.

Fix $y\in[x_{0},r]$. For any $0\leq x\leq y$, 
\begin{eqnarray}
L^{*}(x,y) & = & (1+K)\psi(x)-K\psi(y)-\psi'(x_{0}+)x-\psi(x_{0})+x_{0}\psi'(x_{0}+)-\int_{x_{0}}^{y}\psi''(z)(x-z)dz\nonumber \\
 & = & (1+K)(\psi(x)-\psi(y))+(y-x)\psi'(y).\label{eqn:Lstar7}
\end{eqnarray}
We have $L^{*}(\phi(y),y)=(1+K)(\theta(y)-\psi(y))+(y-\phi(y))\psi'(y)$
and then by (\ref{eq:thm-theta}) and (\ref{eq:thm-eq2}), $L^{*}(\phi(y),y)=0$.
It is also clear that $L^{*}(y,y)=0$.

Differentiating (\ref{eqn:Lstar7}) with respect to $x$, 
\[
\frac{\partial L^{*}}{\partial x}(x,y)=(1+K)\psi'(x)-\psi'(y);\qquad\frac{\partial^{2}L^{*}}{\partial x^{2}}(x,y)=(1+K)\psi''(x).
\]
Then $\frac{\partial^{2}L^{*}}{\partial x^{2}}(x,y)<0$ on $(0,x_{0})$
and $\frac{\partial^{2}L^{*}}{\partial x^{2}}(x,y)>0$ on $(x_{0},y)$.
Since $\frac{\partial L^{*}}{\partial x}(\phi(y),y)=(1+K)\psi'(\phi(y))-\psi'(y)=(1+K)\frac{\psi'(y)}{K+1}-\psi'(y)=0$
and $\frac{\partial L^{*}}{\partial x}(y,y)=K\psi'(y)>0$, it follows
that $\frac{\partial L^{*}}{\partial x}(x,y)>0\textrm{ if }x\in(0,\phi(y)),$
$\frac{\partial L^{*}}{\partial x}(x,y)<0\textrm{ if }x\in(\phi(y),\tilde{x})$
and $\frac{\partial L^{*}}{\partial x}(x,y)>0\textrm{ if }x\in(\tilde{x},y)$,
where $\tilde{x}\in(x_{0},y)$ is such that $\frac{\partial L^{*}}{\partial x}(x,y)|_{x=\tilde{x}}=0$.
It follows that $L^{*}(x,y)<0$ for $x\in[0,\phi(y))\cup(\phi(y),y)$.

Now fix $y>r$. For any $0\leq x\leq y$, and writing $\tilde{\psi}(x)=\psi(x)-x/r$,
\begin{align*}
L^{*}(x,y) & =(1+K)\psi(x)-K-\psi'(x_{0}+)x-\psi(x_{0})+x_{0}\psi'(x_{0}+)-\int_{x_{0}}^{r}\psi''(z)(x-z)dz\\
 & =(1+K)(\psi(x)-1)+(r-x)\psi'(r-)=(1+K)\left(\psi(x)-\frac{x}{r}\right)=(1+K)\tilde{\psi}(x).
\end{align*}
If $x\in(r,y]$ then $L^{*}(x,y)=(1+K)\frac{1}{r}[r-x]<0$. Now suppose
$x\in(0,r)$. Since $\phi(r)=0$, we have $\psi'(0+)=\frac{1}{r}$
and thus $\tilde{\psi}'(0+)=0$. Further $\tilde{\psi}'(r-)=\frac{K}{r}>0$.
Then by the sign of $\psi''(x)$, we get $\tilde{\psi}'(x)$ is negative
and then positive on $(0,r)$. Since $\tilde{\psi}(0)=\tilde{\psi}(r)=0$,
we deduce that $\tilde{\psi}(x)<0$ on $(0,r)$. Thus $L^{*}(x,y)<0$
for $x\in(0,r)$.

From above analysis, we know $L^{*}(x,y)\leq0$ for any $(x,y)$ such
that $0\leq x\leq y$ and $y\geq x_{0}$. This means that $\forall\nu\in\mathcal{E}(x_{0})$
\begin{align*}
\mathcal{L}_{G^{*}}(\nu;\lambda^{*},\gamma^{*},\eta^{*}) & =\int_{0}^{\infty}\int_{0}^{\infty}L^{*}(x,y)\nu(dx,dy)+\lambda^{*}x_{0}+\gamma^{*}\\
 & =\int_{0}^{\infty}\int_{0}^{\infty}L^{*}(x,y)\nu(dx,dy)+\psi(x_{0})\leq\psi(x_{0})=\mathcal{L}_{G^{*}}(\nu^{*};\lambda^{*},\gamma^{*},\eta^{*}).
\end{align*}
Thus $\nu^{*}$ is a symmetric, atom-free Nash equilibrium from Proposition
\ref{thm:Past-regret}. \end{proof}

It remains to show that there exists a constant $r$ and functions
$(\phi,\theta,\psi)$ which satisfy the hypotheses of the Theorem~\ref{thm:Past-regret}
and hence that a symmetric atom-free Nash equilibrium always exists.
The following lemma is key in defining the appropriate entities.

\begin{lem}\label{lem:Past-regret-simplified} Let $J(u)$ solve
the ordinary differential equation 
\begin{equation}
J'(u)=\frac{J(u)+1-(1-u)^{1/(n-1)}}{(K+1)\left[1-u-J(u)^{n-1}\right]}\label{eq:thm-2-J(u)}
\end{equation}
subject to $J(0)=0$ and $u\geq0$. Let $u^{*}=\sup\left\{ u:J(u)<(1-u)^{1/(n-1)}\right\} $.

i) Define 
\[
H(z)=\frac{K}{(K+1)\left[z-J(1-z)^{n-1}\right]}
\]
on $[z^{*},1]$, where $z^{*}=1-u^{*}$. Then $z^{*}>0$, $H$ is
positive on $(z^{*},1)$ and $\int_{z^{*}}^{1}\exp\left(\int_{w}^{1}H(v)dv\right)dw<(K+1)$.

ii) Define 
\[
r=\frac{x_{0}(K+1)}{(K+1)-\int_{z^{*}}^{1}\exp\left(\int_{w}^{1}H(v)dv\right)dw}
\]
and 
\[
\Psi(z)=\frac{r}{K+1}\left[(K+1)-\int_{z}^{1}\exp\left(\int_{w}^{1}H(v)dv\right)dw\right]
\]
on $[z^{*},1]$. Let $\psi=\Psi^{-1}$ be the inverse function of
$\Psi$. Then $x_{0}<r<\infty$ and $\psi:[x_{0},r]\mapsto[0,1]$
is a strictly increasing and strictly convex function that satisfies
$\psi(r)=1$, $\psi'(r-)=\frac{K+1}{r}$ and $\psi''(r-)=\frac{K(K+1)}{r^{2}}$.

iii) Define 
\begin{equation}
\phi(y)=y-\frac{K\psi'(y)}{\psi''(y)}.\label{eq:thm-2-phi}
\end{equation}
Then $\phi:[x_{0},r]\mapsto[0,x_{0}]$ is a strictly decreasing function
with $\phi(x_{0})=x_{0}$.

iv) Define 
\[
\theta(y)=\psi(x_{0})+\frac{1}{K+1}\int_{x_{0}}^{y}\phi'(z)\psi'(z)dz
\]

Then $\theta:[x_{0},r]\mapsto[0,1]$ is a strictly decreasing function
with $\theta(x_{0})=\psi(x_{0})$. Moreover, $\theta(y)=\psi(y)-(y-\phi(y))\psi'(y)/(K+1)$.

\end{lem}

\begin{proof} It is easily seen that $J(u)$ is a strictly increasing
function at least until $J(u)=\sqrt[n-1]{1-u}$, and that $u^{*}<1$.

i) Since $u^{*}<1$, $z^{*}=1-u^{*}>0$. Since $J$ is increasing,
for any $u\in(0,u^{*})$ $J(u)^{n-1}\leq J(u^{*})^{n-1}$ and thus
$1-u-J(u)^{n-1}\geq1-u-J(u^{*})^{n-1}=u^{*}-u$. This means for any
$z\in(z^{*},1)$ we have $z-J(1-z)^{n-1}\geq z-z^{*}>0$ and then
$0<H(z)\leq\frac{K}{(K+1)(z-z^{*})}$. Moreover, 
\begin{align*}
\int_{z^{*}}^{1}\exp\left(\int_{w}^{1}\frac{K}{(K+1)(v-z^{*})}dv\right)dw & =\int_{z^{*}}^{1}\exp\left(\frac{K}{K+1}\ln\frac{1-z^{*}}{w-z^{*}}\right)dw=\int_{z^{*}}^{1}\left(\frac{1-z^{*}}{w-z^{*}}\right)^{\frac{K}{K+1}}dw\\
 & =(1-z^{*})^{\frac{K}{K+1}}(1-z^{*})^{\frac{1}{K+1}}(K+1)=(1-z^{*})(K+1)<(K+1),
\end{align*}
and it follows that $\int_{z^{*}}^{1}\exp\left(\int_{w}^{1}H(v)dv\right)dw<(K+1)$.

ii) Since $0<(K+1)-\int_{z^{*}}^{1}\exp\left(\int_{w}^{1}H(v)dv\right)dw<(K+1)$,
we have that $x_{0}<r<\infty$. Taking derivatives of $\Psi$ on $(z^{*},1)$,
we find $\Psi'(z)=\frac{r}{K+1}\exp\left(\int_{z}^{1}H(v)dv\right)>0$
and $\Psi''(z)=-H(z)\Psi'(z)<0$. Then since $\Psi(z^{*})=x_{0}$
and $\Psi\left(1\right)=r$, $\Psi$ is a strictly increasing and
strictly concave function from $[z^{*},1]$ to $[x_{0},r]$. Thus
$\psi=\Psi^{-1}:[x_{0},r]\mapsto[0,1]$ is a strictly increasing and
strictly convex function satisfying $\psi(r)=1$.

Moreover, $\psi'(y)=\frac{1}{\Psi'(\psi(y))}$ and thus $\psi''(y)=-\frac{\Psi''(\psi(y))\psi'(y)}{\Psi'(\psi(y))^{2}}$.
Then since $\Psi(1)=r$, $\Psi'(1-)=\frac{r}{K+1}$ and $\Psi''(1-)=-\frac{rK}{(K+1)^{2}}$,
we get $\psi(r)=1$, $\psi'(r-)=\frac{1}{\Psi'(1-)}=\frac{K+1}{r}$
and $\psi''(r-)=-\frac{\Psi''(1-)\psi'(r-)}{\Psi'(1-)^{2}}=\frac{K(K+1)}{r^{2}}$.

iii) Letting $u=1-z$ in (\ref{eq:thm-2-J(u)}), we get $J'(1-z)=\frac{J(1-z)+1-z^{1/(n-1)}}{(K+1)\left[z-J(1-z)^{n-1}\right]}=\frac{H(z)}{K}\left(J(1-z)+1-z^{\frac{1}{n-1}}\right)$.
Thus, 
\begin{align}
H'(z) & =-\frac{K\left[1+(n-1)J(1-z)^{n-2}J'(1-z)\right]}{(K+1)\left[z-J(1-z)^{n-1}\right]^{2}}\nonumber \\
 & =-\frac{K+1}{K}H(z)^{2}\left[1+(n-1)J(1-z)^{n-2}\frac{H(z)}{K}\left(J(1-z)+1-z^{\frac{1}{n-1}}\right)\right]\nonumber \\
 & =-\frac{K+1}{K}H(z)^{2}-\frac{(K+1)(n-1)}{K^{2}}H(z)^{3}\left[J(1-z)^{n-1}+J(1-z)^{n-2}\left(1-z^{\frac{1}{n-1}}\right)\right]\label{eq:thm-2-H'(z)-J}
\end{align}
It is clear from the middle line above that $\frac{H'(z)}{H(z)^{2}}<-\frac{K+1}{K}$.

Note that $\phi'(y)=1-K\left(1-\frac{\psi'(y)\psi'''(y)}{\psi''(y)^{2}}\right)$.
Since $H(z)=-\frac{\Psi''(z)}{\Psi'(z)}=\frac{\psi''(\Psi(z))}{\psi'(\Psi(z))^{2}}$,
we have that 
\begin{equation}
H(\psi(y))=\frac{\psi''(y)}{\psi'(y)^{2}}\label{eq:thm-2-H(phi(y))}
\end{equation}
and thus $H'(\psi(y))\psi'(y)=\frac{\psi'''(y)}{\psi'(y)^{2}}-\frac{2\psi''(y)^{2}}{\psi'(y)^{3}}$.
Thus $1-\frac{\psi'(y)\psi'''(y)}{\psi''(y)^{2}}=-H'(\psi(y))\frac{\psi'(y)^{4}}{\psi''(y)^{2}}-1$
and therefore 
\begin{equation}
\phi'(y)=K\frac{H'(\psi(y))}{H(\psi(y))^{2}}+K+1.\label{eq:thm-2-phi_prime}
\end{equation}
It follows that $\phi'(y)<0$.

Since $\psi'(r-)=\frac{K+1}{r}$, $\psi'(y)>0$ and $\psi''(y)>0$
on $(x_{0},r)$, we know $\psi'(x_{0}+)$ is bounded. Then, since
$\frac{\psi''(x_{0}+)}{\psi'(x_{0}+)^{2}}=H(\psi(x_{0}))=H(z^{*})=+\infty$,
we get $\psi''(x_{0}+)=+\infty$. Substituting these values into (\ref{eq:thm-2-phi})
we obtain $\phi(x_{0})=x_{0}$.

iv) The statements about $\theta$ are either trivial, or follow as
in the derivation of (\ref{eq:10}).

\end{proof}

\begin{thm}\label{thm:Past-regret-simplified} Let $r$, $\psi$,
$\phi$, $\theta$ be as defined in Lemma~\ref{lem:Past-regret-simplified}.

Then there exists a symmetric, atom-free Nash equilibrium for the
problem for which $X_{\tau^{i}}^{i}$ has distribution $F$ where
$F(x)=0$ for $x\leq0$, $F(x)=1$ for $x\geq r$ and otherwise 
\[
F(x)=\begin{cases}
\theta(\phi^{-1}(x))^{\frac{1}{n-1}} & \textrm{ if }0<x<x_{0},\\
\psi(x)^{\frac{1}{n-1}} & \textrm{ if }x_{0}\leq x<r.
\end{cases}
\]
\end{thm}

\begin{proof} By (\ref{eq:thm-2-phi}) and (\ref{eq:thm-2-H(phi(y))}),
$\frac{1}{H(\psi(y))}=\frac{\psi'(y)^{2}}{\psi''(y)}=\frac{y-\phi(y)}{K}\psi'(y)$.
Hence $\theta(y)=\psi(y)-\frac{K}{(K+1)H(\psi(y))}=J(1-\psi(y))^{n-1}.$
Thus letting $z=\psi(y)$ in (\ref{eq:thm-2-H'(z)-J}), 
\[
H'(\psi(y))=-\frac{K+1}{K}H(\psi(y))^{2}-\frac{(K+1)(n-1)}{K^{2}}H(\psi(y))^{3}\left[\theta(y)+\left(1-\psi(y)^{\frac{1}{n-1}}\right)\theta(y)^{\frac{n-2}{n-1}}\right],
\]
and using (\ref{eq:thm-2-phi_prime}) and rearranging above equation,
\[
\frac{y-\phi(y)}{n-1}\frac{\phi'(y)\psi'(y)}{K+1}=-\left[\theta(y)+\left(1-\psi(y)^{\frac{1}{\left(n-1\right)}}\right)\theta(y)^{\frac{n-2}{n-1}}\right].
\]
Then substituting (\ref{eq:thm-eq1}) into above equation, (\ref{eq:thm-eq3})
follows. Therefore, $(\dagger)$ holds. Then using Theorem \ref{thm:Past-regret},
we obtain the symmetric, atom-free Nash equilibrium given in Theorem
\ref{thm:Past-regret-simplified}. \end{proof}

\begin{exa} As an example we consider a 3-player contest. Set $x_{0}=1$
and $n=3$. In Figures \ref{fig:Plot-of-G} and \ref{fig:Plot-of-g}
we give graphs of the optimal distribution $G^{*}(x)$ and its density
function $g^{*}(x)$ for various values of $K$.



Figure 1 here.



Figure 2 here

As we can seen in Figure \ref{fig:Plot-of-G}, the right endpoint
$r=r(K)$ of $G^{*}(x)$ decreases as $K$ increases. Moreover $r(K)$
tends to $n=3$ as $K$ decreases to $0$ and tends to $x_{0}$ as
$K$ increases to $+\infty$. We also find that $G^{*}(x)$ tends
to the equilibrium distribution of the original contest as $K$ decreases
to $0$ and $G^{*}(x)$ tends to the Heaviside function $\mathcal{H}_{\{x\geq x_{0}\}}$
as $K$ increases to $+\infty$. From Figure \ref{fig:Plot-of-g},
we find $g^{*}(x)$ jumps at $x_{0}$ if $K>0$ and $g^{*}(x)$ tends
to $+\infty$ as $y$ tends to $0$.

Intuitively, if $K$ is very large then the player does not aim for
large values of the stopped process, for then she risks a moderate
value of the maximum together with a small and losing value for the
stopped process. Because of the large penalty she wishes to avoid
such outcomes.

\end{exa}

\begin{exa} In the 2-player contest we can give explicit expressions
for several of the quantities of interest.

Set $n=2$. Substituting (\ref{eq:thm-eq1}), (\ref{eq:thm-eq2})
and (\ref{eq:thm-theta}) into (\ref{eq:thm-eq3}), we get 
\[
(y-\phi(y))\frac{1}{1+K}\phi'(y)\psi'(y)=(\psi(y)-1)-\psi(y)+\frac{K}{K+1}\frac{y-\phi(y)}{K}\psi'(y).
\]
Defining $\varphi(y)=y-\phi(y)$ the above equation simplifies to
$\varphi(y)\varphi'(y)=\frac{K+1}{\psi'(y)}$. Differentiating this
expression and using (\ref{eq:thm-eq2}) we have 
\[
\left[\varphi(y)\varphi'(y)\right]'=-(K+1)\frac{\psi''(y)}{\psi'(y)^{2}}=-(K+1)\frac{K}{\varphi(y)\psi'(y)}=-K\varphi'(y),
\]
and then 
\begin{equation}
\varphi(y)\varphi'(y)=-K\varphi(y)+K\varphi(r)+\varphi(r)\varphi'(r-).\label{eq:2-player-varphi-varphiprime-C}
\end{equation}
Since $\varphi(r)=r$ and $\psi'(r-)=\frac{K+1}{r}$ we have $\varphi'(r-)=1$.
Then (\ref{eq:2-player-varphi-varphiprime-C}) becomes $\varphi(y)\varphi''(y)=-K\varphi(y)+r(K+1),$
which using the boundary condition $\varphi(r)=r$ has solution 
\[
r-y=\frac{r(K+1)}{K^{2}}\ln\left((K+1)-\frac{K\varphi(y)}{r}\right)-\frac{r-\varphi(y)}{K}
\]
Using $\varphi(x_{0})=x_{0}-\phi(x_{0})=0$, we find 
\[
r=x_{0}\frac{K^{2}}{(K+1)\left[K-\ln\left(1+K\right)\right]}
\]
and therefore the implicit form of $\varphi(y)$ for $y\in[x_{0},r]$
is 
\[
y=x_{0}-\frac{\varphi(y)}{K}-\frac{x_{0}}{K-\ln\left(1+K\right)}\ln\left[1-\varphi(y)\frac{K-\ln\left(1+K\right)}{Kx_{0}}\right]
\]
and $\phi(y)=y-\varphi(y)$. It is possible to express $\psi$ and
$\theta$ in terms of $\varphi$, and thence the optimal distribution
$G{}^{*}$ of $X_{\tau^{i}}^{i}$ and the optimal conditional distribution
of $M_{\tau^{i}}^{i}$ given $X_{\tau^{i}}^{i}$, but these expressions
are not so compact. \end{exa}

\section{Contest with regret over failure to stop at the best time\label{sec:All-regrets}}

In this section we discuss the contest with regret over failure to
stop at the best time. A player experiences regret if she could have
won if she had stopped at the maximum value over the whole path. The
maximum value $M^{i}$ is given by

\[
M^{i}:=M_{H_{0}^{i}}^{i}=\sup_{0\leq t\leq H_{0}^{i}}X_{t}^{i}.
\]

\begin{thm} \label{thm:All-regret} There exists a symmetric, atom-free
Nash equilibrium for the problem for which $X_{\tau^{i}}^{i}$ has
law $F(x)$, where for $x\geq0$ 
\[
F(x)=\min\left\{ \sqrt[n-1]{\frac{x}{nx_{0}}},1\right\} .
\]
\end{thm}

\begin{proof} The agent's expected payoff is 
\[
(1+K)\mathbb{E}[F(X_{\tau^{i}}^{i})^{n-1}]-K\mathbb{E}[F(M_{H_{0}^{i}}^{i})^{n-1}]
\]
But the latter term is independent of the stopping rule used by the
agent. Hence, in determining her optimal strategy the agent need only
consider $(1+K)\mathbb{E}[F(X_{\tau^{i}}^{i})^{n-1}]$. Modulo the
initial constant, this is the same objective function as in the standard
case. \end{proof}

\begin{remark}The agent follows exactly the same Nash equilibrium
strategy as an agent in the original contest, in which there is no
penalty. The intuition behind is that the regret is determined by
$M_{H_{0}^{i}}^{i}$ but player cannot change the distribution of
$M_{H_{0}^{i}}^{i}$ by changing the choices of stopping time $\tau^{i}$.\end{remark}

\section{Derivation of the equilibrium distribution\label{sec:Derivation}}

This section is intended to illustrate how we derived the optimal
multipliers and the candidate Nash equilibrium in Sections \ref{sec:Without-regrets},
\ref{sec:Future-regrets} and \ref{sec:Past-regrets} and also the
boundary conditions in Section \ref{sec:Past-regrets}. The Lagrangian
approach gives a general method for finding the optimal solution,
which is distinct from the ideas in Seel and Strack~\cite{GamblinginContests},
and can be generalised to other settings.

\subsection{Contest without regret or with regret over stopping too soon}

Recall the definition of the Lagrangian $\mathcal{L}_{F}(G;\lambda,\gamma)$
for the optimization problem (\ref{eq:Optimizationproblemnoregretssymmetric})
and (\ref{eq:Optimization problem future failure symmetric}). (We
cover the more complicated case of regret from continuing beyond a
winning time in a separate section.)

Denote by $L_{F}(x)$ by the integrand in $\mathcal{L}_{F}$, so that
$\mathcal{L}_{F}(G;\lambda,\gamma)=\int_{0}^{\infty}L_{F}(x)G(dx)+\lambda x_{0}+\gamma$.
In order to have a finite optimal solution we require $L_{F}(x)\leq0$
on $[0,\infty)$. Let $\mathcal{D}_{F}$ be the set of $(\lambda,\gamma)$
such that $\mathcal{L}_{F}(\cdot;\lambda,\gamma)$ has a finite maximum.
Then $\mathcal{D}_{F}$ is defined by 
\[
\mathcal{D}_{F}=\{(\lambda,\gamma):L_{F}(x)\leq0\textrm{ on }[0,\infty)\}.
\]
In order to reach the maximum value, we require $G(dx)=0$ when $L_{F}(x)<0$.
This means that for $(\lambda,\gamma)\in\mathcal{D}_{F}$ the maximum
of $\mathcal{L}_{F}(\cdot;\lambda,\gamma)$ occurs at $G^{*}$ such
that $G^{*}(dx)=0$ when $L_{F}(x)<0$. Conversely we expect that
$G^{*}(dx)>0$ when $L_{F}(x)=0$. If the Nash equilibrium is symmetric
then we must have $G^{*}(x)=F(x)$ and then $L_{G^{*}}(x)\leq0$,
and $L_{G^{*}}(x)=0$ when $G^{*}(dx)>0$. Introduce $a=\inf\{x:G^{*}(x)>0\}$
and $b=\sup\{x:G^{*}(x)<1\}$ which are the limits on the support
of $G^{*}$.

\subsubsection{Contest without regret}

For the optimization problem (\ref{eq:Optimizationproblemnoregretssymmetric}),
$L_{F}(x)=F(x)^{n-1}-\lambda x-\gamma$. Observe that $0\leq F\left(0\right)^{n-1}$
so that if $(\lambda,\gamma)\in\mathcal{D}_{F}$ then $\gamma$ is
non-negative.

Since $L_{G^{*}}(x)=G^{*}(x)^{n-1}-\lambda x-\gamma$, we must have
$G^{*}(x)=\sqrt[n-1]{\lambda x+\gamma}$ at least when $G^{*}(dx)>0$.
Since we are searching for atom-free solutions we must have $G^{*}(x)=\sqrt[n-1]{\lambda x+\gamma}$
on the whole of the interval $[a,b]$. Moreover, since $G^{*}$ is
non-decreasing and not constant we must have $\lambda>0$.

Since $G^{*}$ is atom-free, $G^{*}(a)=0$ and hence $\lambda a+\gamma=0$.
Then by the non-negativity of $a$ and $\gamma$ and the positivity
of $\lambda$ it follows that $\gamma=0=a$. Thus $G^{*}(x)=\sqrt[n-1]{\lambda x}$
on $[0,b]$ for some $\lambda$ and $b$ which we must find.

For a feasible solution, $\int_{0}^{\infty}G^{*}(dx)=1$ and $\int_{0}^{\infty}xG^{*}(dx)=x_{0}$,
so that 
\[
1=\int_{0}^{b}d\left(\sqrt[n-1]{\lambda x}\right)=\sqrt[n-1]{\lambda b}.\hspace{10mm}x_{0}=\int_{0}^{b}xd\left(\sqrt[n-1]{\lambda x}\right)=\frac{\sqrt[n-1]{\lambda}}{n}b^{\frac{n}{n-1}}=\sqrt[n-1]{\lambda b}\frac{b}{n};
\]
Hence $b=nx_{0}$ and $\lambda=1/(nx_{0})$. This gives us that $G^{*}$
is the distribution function given in Theorem \ref{thm:no-regret}.

\subsubsection{Contest with regret from stopping too soon}

Now we have that $L_{F}\left(x\right)=(1+K)F(x)^{n-1}-Kx\int_{x}^{\infty}\frac{F(y)^{n-1}}{y^{2}}dy-\lambda x-\gamma.$
Then $L_{F}\left(0\right)=(1+K)F(0)^{n-1}-\gamma$ and as before,
if $(\lambda,\gamma)\in\mathcal{D}_{F}$ then $\gamma$ is non-negative.

Let $\psi(x)=G^{*}(x)^{n-1}$ then $L_{G^{*}}(x)$ becomes

\begin{equation}
L_{G^{*}}(x)=(1+K)\psi(x)-Kx\int_{x}^{\infty}\frac{\psi(y)}{y^{2}}dy-\lambda x-\gamma.\label{eq:psi future failure}
\end{equation}
Thus we expect $\psi(x)$ is the solution to $L_{G^{*}}(x)=0$ at
least when $\psi(dx)>0$.

Setting $L_{G^{*}}(x)=0$ and differentiating (\ref{eq:psi  future failure})
twice with respect to $x$, we find 
\[
(1+K)\psi''(x)x+K\psi'(x)=0.
\]
Thus $\psi(x)=C_{1}x^{\frac{1}{K+1}}+C_{2}$, where $C_{1}$ and $C_{2}$
are some constants, and then $G^{*}(x)=\sqrt[n-1]{C_{1}x^{\frac{1}{K+1}}+C_{2}}$
at least when $G^{*}(dx)>0$. Since we are seeking an atom-free solution
we must have $G^{*}(x)=\sqrt[n-1]{C_{1}x^{\frac{1}{K+1}}+C_{2}}$
on the whole interval of $[a,b]$, where $C_{1}>0$.

Substituting $\psi(x)=G^{*}(x)^{n-1}=(C_{1}x^{\frac{1}{K+1}}+C_{2})\wedge1$
into (\ref{eq:psi future failure}), and setting $L_{G^{*}}(x)=0$
we have $\forall x\in[a,b]$ 
\begin{align*}
0 & =(1+K)\left(C_{1}x^{\frac{1}{K+1}}+C_{2}\right)-Kx\int_{x}^{b}\frac{C_{1}y^{\frac{1}{K+1}}+C_{2}}{y^{2}}dy-Kx\int_{b}^{\infty}\frac{1}{y^{2}}dy-\lambda x-\gamma\\
 & =\left[(1+K)C_{1}b^{\frac{1}{K+1}}+KC_{2}-K-\lambda b\right]\frac{x}{b}+C_{2}-\gamma.
\end{align*}
This gives us optimal multipliers $\gamma^{*}=C_{2}$ and $\lambda^{*}=\frac{1}{b}\left[(1+K)C_{1}b^{\frac{1}{K+1}}+KC_{2}-K\right].$

Since $G^{*}$ is atom-free, $G^{*}(a)=0$ and hence $C_{1}a^{\frac{1}{K+1}}+C_{2}=0$,
and from the non-negativity of $a$ and $\gamma^{*}=C_{2}$ and the
positivity of $C_{1}$ it follows that $C_{2}=a=0$. Thus $G^{*}(x)=\sqrt[n-1]{C_{1}x^{\frac{1}{K+1}}}$
on $[0,b]$ for some $C_{1}$ and $b$ which can be identified using
the fact that $G^{*}$ corresponds to a probability distribution with
mean $x_{0}$. In particular, setting $N=1+(K+1)(n-1)$ for a feasible
solution, 
\[
\begin{cases}
1=\int_{0}^{b}d\left(\sqrt[n-1]{C_{1}x^{\frac{1}{K+1}}}\right)=\sqrt[n-1]{C_{1}b^{1/(K+1)}},\\
x_{0}=\int_{0}^{b}xd\left(\sqrt[n-1]{C_{1}x^{\frac{1}{K+1}}}\right)=\frac{\sqrt[n-1]{C_{1}}}{(K+1)(n-1)+1}b^{\frac{\left(K+1\right)\left(n-1\right)+1}{\left(K+1\right)\left(n-1\right)}}=\sqrt[N-1]{C_{1}^{K+1}b}\frac{b}{N}.
\end{cases}
\]
Hence $C_{1}=b^{-1/(K+1)}$ and then $b=Nx_{0}$ and $C_{1}=\sqrt[K+1]{\frac{1}{Nx_{0}}}$.
Thus $G^{*}(x)=\sqrt[N-1]{x/Nx_{0}}$ on $[0,Nx_{0}]$.

\subsection{Contest with regret over past failure}

Recall the definition of the Lagrangian $\mathcal{L}_{F}(\nu;\lambda,\gamma,\eta)$
for the optimization problem of Section~\ref{sec:Past-regrets}.
Let $L_{F}\left(x,y\right)$ be the integrand in the definition of
$\mathcal{L}_{F}$ as given in (\ref{eqn:LagrangianP}). In order
to have a finite optimal solution we require $L_{F}(x,y)\leq0$ on
$[0,\infty)\times[x_{0},\infty)$. Let $\mathcal{D}_{F}$ be the set
of $(\lambda,\gamma,\eta)$ such that $\mathcal{L}_{F}(\cdot;\lambda,\gamma,\eta)$
has a finite maximum. Then $\mathcal{D}_{F}$ is defined by 
\[
\mathcal{D}_{F}=\{(\lambda,\gamma,\eta):L_{F}(x,y)\leq0;x\geq0,y\geq x_{0}\}.
\]
For $(\lambda,\gamma,\eta)\in\mathcal{D}_{F}$ the maximum of $\mathcal{L}_{F}(\cdot;\lambda,\gamma,\eta)$
occurs at a measure $\nu^{*}$ such that $\nu^{*}(dx,dy)=0$ when
$L_{F}(x,y)<0$. Conversely we expect that $\nu^{*}(dx,dy)>0$ when
$L_{F}(x,y)=0$.

Let $G^{*}(x)=\nu^{*}(\{(u,y):u\leq x,x_{0}\leq y<\infty\})$ be the
marginal of $\nu^{*}$. If the Nash equilibrium is symmetric then
we must have $G^{*}(x)=F(x)$ and $L_{G^{*}}(x,y)=0$ when $\nu^{*}(dx,dy)>0$.
Motivated by the results of previous sections we expect $G^{*}$ to
place mass on an interval $[a,b]$ where $0=a<x_{0}<b$. In this section
we write $b=r$ for the upper limit.

It follows from the discussion before Theorem~\ref{thm:Past-regret}
that for an optimal solution either $X_{\tau^{i}}^{i}=M_{\tau^{i}}^{i}$
or $X_{\tau^{i}}^{i}=\phi(M_{\tau^{i}}^{i})$ for some decreasing
function $\phi$. Hence, for $x_{0}\leq y\leq r$ we expect $\nu^{*}(dx,dy)>0$
if and only if either $x=y$ or $x=\phi(y)$. Let $\psi(x)=G^{*}(x)^{n-1}$.
Then $L_{G^{*}}(x,y)$ becomes 
\[
L(x,y):=L_{G^{*}}(x,y)=(1+K)\psi(x)-K\psi(y)-\lambda x-\gamma-\int_{x_{0}}^{y}\eta(z)(x-z)dz.
\]
Fixing $y\in(x_{0},r)$, and using $L(x,y)\leq0$ for any $0\leq x\leq y$,
together with $L(\phi(y),y)=0$ we expect $\frac{\partial L}{\partial x}(\phi(y),y)=0$.

Thus $\forall y\in(x_{0},r)$, $\psi$ and $\phi$ must solve \begin{numcases}{}
L(y,y)=\psi(y)-\lambda y-\gamma-\int_{x_{0}}^{y}\eta(z)(y-z)dz=0,\label{eq:L(y,y)}\\ 
L(\phi(y),y)=(1+K)\psi(\phi(y))-K\psi(y)-\lambda\phi(y)-\gamma-\int_{x_{0}}^{y}\eta(z)(\phi(y)-z)dz=0,\label{eq:L(phi,y)}\\ 
\frac{\partial L}{\partial x}(\phi(y),y)  =(1+K)\psi'(\phi(y))-\lambda-\int_{x_{0}}^{y}\eta(z)dz=0.\label{eq:dL/dx} 
\end{numcases} Differentiating (\ref{eq:L(y,y)}) with respect to $y$, yields 
\begin{equation}
\psi'(y)-\lambda-\int_{x_{0}}^{y}\eta(z)dz=0.\label{eq:entasolved}
\end{equation}
Comparing (\ref{eq:dL/dx}) with (\ref{eq:entasolved}), we find $\psi'(y)=(1+K)\psi'(\phi(y)).$
If we now set $\theta(y)=\psi(\phi(y))$, then (\ref{eq:thm-eq1})
follows. From (\ref{eq:entasolved}) we find 
\begin{equation}
\psi''(y)-\eta(y)=0,\label{eq:entaRepresentation}
\end{equation}
Then, differentiating (\ref{eq:L(phi,y)}) with respect to $y$, and
using (\ref{eq:dL/dx}) we obtain $-K\psi'(y)-\eta(y)(\phi(y)-y)=0$,
and (\ref{eq:thm-eq2}). Finally, (\ref{eq:thm-eq3}) comes directly
from (\ref{eqn:PerkinsHPdefnphiD}) on noting that $G(\phi(m))=\theta(m)^{1/(n-1)}$.

Next we deduce the boundary conditions. First note that from (\ref{eqn:PerkinsHPdefnphiI})
we can infer that $\phi(x_{0})=x_{0}$ and $\phi(r)=0$. Hence $\theta(x_{0})=\psi(x_{0})$
and $\theta(r)=0$. Given that (\ref{eq:thm-eq1}) and (\ref{eq:thm-eq2})
hold, as in the proof of Theorem \ref{thm:Past-regret}, we have that
(\ref{eq:10}) holds. Letting $y=r$ and using $\psi(r)=1$, $\phi(r)=0$
and $\theta(r)=0$, we find $0=\frac{-r}{K+1}\psi'(r-)+\psi(r)$ and
hence $\psi''(r-)=\frac{K+1}{r}$. Further, letting $y=r$ in (\ref{eq:thm-eq2})
we get $\psi''(r-)=K(K+1)/r^{2}$ as required.

Lastly, we derive the optimal multipliers which we write as $\eta^{*}$,
$\lambda^{*}$ and $\gamma^{*}$. From (\ref{eq:entaRepresentation}),
$\eta^{*}(y)=\psi''(y)$ for $y\in(x_{0},r)$. Then, from (\ref{eq:entasolved}),
$\lambda^{*}=\psi'(y)-\int_{x_{0}}^{y}\eta^{*}(z)dz=\psi'(y)-\int_{x_{0}}^{y}\psi''(z)dz=\psi'(x_{0}+)$.
Finally (\ref{eq:L(y,y)}) yields, 
\[
\gamma^{*}=\psi(y)-\lambda y-\int_{x_{0}}^{y}\eta^{*}(z)(y-z)dz=\psi(y)-y\psi'(x_{0}+)-\int_{x_{0}}^{y}\psi''(z)(y-z)dz=\psi(x_{0})-x_{0}\psi'(x_{0}+).
\]

\bibliographystyle{plain}
\bibliography{gambling}

\end{document}